\documentclass[12pt]{amsart}

\usepackage[onehalfspacing]{setspace}
\usepackage[foot]{amsaddr}

\usepackage{natbib}
\usepackage{comment}

\def\la{\lambda}
\def\al{\alpha}

\def\ta{\theta}
\def\g{\gamma}

\def\then{\Longrightarrow}

\def\Re{\mathbf{R}}

\def\cala{\mathcal{A}}
\def\T{\mathcal{S}}

\newcommand{\df}[1]{\textit{#1}}

\theoremstyle{plain}
\newtheorem{theorem}{Theorem}
\newtheorem{lemma}[theorem]{Lemma}
\newtheorem{proposition}[theorem]{Proposition}
\newtheorem{corollary}[theorem]{Corollary}
\theoremstyle{remark}
\newtheorem{remark}[theorem]{Remark}

\begin{document}

\title[Statistical discrimination]{A characterization of ``Phelpsian''
  statistical discrimination}

\thanks{
Echenique thanks the NSF for support through the grants
 SES-1558757 and CNS-1518941. We are grateful to Leeat Yariv for comments
    on a previous draft.}

\author[Chambers]{Christopher P. Chambers}
\author[Echenique]{Federico Echenique}

\address[Chambers]{Department of Economics, Georgetown University}
\address[Echenique]{Division of the Humanities and Social Sciences,
  California Institute of Technology}

\begin{abstract}We establish that statistical discrimination is
  possible if and only if it is impossible to uniquely identify the
  signal structure observed by an employer from a realized empirical
  distribution of skills.  The impossibility of statistical
  discrimination is shown to be equivalent to the existence of a fair,
  skill-dependent,  remuneration for workers.   Finally,
  we connect the statistical discrimination literature to Bayesian
  persuasion, establishing that 
  if  discrimination is absent, then the optimal
  signaling problem results in a linear payoff function (as well as a
  kind of converse).\end{abstract}

\maketitle

\section{Introduction}
In seminal contributions, Arrow (\citeyear{arrow71,arrow73}) and
\cite{phelps} postulated that discrimination along racial lines, or
gender identities, can have a statistical explanation. In this note we
focus on Phelps' notion of statistical discrimination: on the idea
that two populations of workers, who are in essence identical, may have different
economic remunerations for purely informational reasons.\footnote{
We follow the interpretation of Phelps' model due to
\cite{aigner1977statistical}. Arrow's theory of statistical
discrimination relies on a coordination failure, and is quite
different from Phelps'.
Statistical discrimination stands in contrast with taste-based
discrimination, as in \cite{becker}.}

Phelps' theory connects worker remuneration with the distribution of
signals that can be observed about worker skills. Phelps
assumes a firm who observes a signal about the
underlying skills of a worker. The firm observes the signal before
assigning the worker to a task. The worker is paid her expected
contribution to the firm, conditional on the firm's observed
signal about the worker. (A competitive market ensures that workers
are paid their   contributions.) Consider now two
populations of workers: group A and group B. If the signal is more
informative for As than for Bs, then (the argument goes), a
 worker from group A may be ex-ante more valuable to the firm than a B
 worker. This is because the additional information about the A worker
 may be used to better assign her a task matching her skills.  Even
 more, the signal may be the result of a test that has
been designed with a  population from group A in mind. The signal implemented
by the test will then be more informative about the skills of a
prospective A worker than a B worker.\footnote{As an example, Aigner
  and Cain cite
  evidence from the education literature to the effect that the SAT is
  less informative about the abilities of African-American students
  than it is for white students.}

As a consequence of the difference in informativeness, the firm may value a
group A worker over a group B worker.
We formulate the theory of statistical discrimination using the
language of the recent literature on informational design. A firm
observes a signal about a worker's skills, and bases both the assigned
task and the payment to the worker on the revenue it expects to gain
from the action taken by the worker at the firm.  A group of workers
comes with a distribution over signals: an information structure. The
distribution over signals of group A may be more informative than the
distribution over signals of group B. We say that statistical discrimination
is present if two groups of workers, each group having their distinct
distribution over signals, but {\em the same distribution of skills,}
receive different payments in expectation.

Our contribution is to connect statistical discrimination with two
seemingly distinct properties of the economic environment: one is
identification (in the
econometric sense) of signals from skills, and the other is the linearity of
firm revenue in ``fair'' skill-dependent payoffs. First, we show that
the absence of statistical discrimination is  equivalent to the econometric
identification of signals. Specifically, we prove that statistical
discrimination is not possible if and only if every given distribution
of skills arises from a unique distribution of signals. By definition,
when discrimination is possible, the identification property must be
violated.  Our contribution is in the converse: whenever
identification is impossible, discrimination can arise.

Second, we  show that identification, and therefore the absence of
discrimination,  is equivalent to the existence of a fair skill-based
remuneration for workers.  Workers' payments are a linear combination
of the fair remunerations.  Each list of skills must be associated
with a value, which is independent of any signals, and every worker is
paid the expectation according to the distribution of skills inherent
in her realized signal.

Finally, we show that the optimal information structure in the sense of
\cite{kamenica2011bayesian} achieves precisely the  fair remuneration
in our results.

\section{The model}\label{sec:model}

\subsection{Notation} A set is \df{binary} it is has one or two
elements. If $A$ is a closed subset of a Euclidean space, we denote
by $\Delta(A)$ the set of Borel probability measures on $A$.

\subsection{The model}

The model involves a firm and a worker. The firm faces uncertainty
over the revenues it can obtain from the worker's actions. The firm's revenue
depends on the worker's skills, and how those skills match up with the
technology of the firm. 

Let $\Theta$ denote a finite set of uncertain \emph{states
  of the world}; these states represent the skill set of the worker,
and are unknown to the firm.  The firm asks the 
worker to undertake some action, and it only cares about the
state-contingent payoff that results from the worker's action.
Formally, then, an \emph{action} is an element $a\in\Re^{\Theta}$.
Thus, the task of the firm is to properly match a worker to an action
with the appropriate skill set.

There is a closed set of \emph{signals}, or \emph{payoff-relevant
  types}, $\T\subseteq \Delta(\Theta)$. Here we identify signals with
the posterior distribution that they induce over $\Theta$. The
firm observes $s\in\T$ before asking the worker to undertake an
action.  Thus, the goal of the firm is to choose the appropriate
action for the appropriate worker, after a signal of worker skill
has been observed.

The firm solves the following problem. For a given $s\in \T$, and
finite set of actions $A$,
\[v_A(s)\equiv\max_{a\in A}\sum_{\theta\in \Theta}a(\theta)s(\theta).\]
Given signal $s\in \T$, $v_A(s)$ is the maximal
expected revenue the employer can achieve. We maintain the
assumption that labor markets are competitive, and therefore a worker
of type $s$ is paid the revenue $v_A(s)$ that she generates for the
firm. This is as in \cite{phelps} and \cite{aigner1977statistical}.
Observe that $v_A$ is the ``value function'' of $A$, as in
\citet{blackwell1953equivalent} or \citet{machina}, and is thus always
convex.

A probability $\pi\in\Delta(\T)$ is an \df{information structure}. It induces a probability over $\Theta$ via: $p_{\pi}(\theta)=\int_{\T} s(\theta)d\pi(s)$.  For a set $E\subseteq \T$, we can interpret $\pi(E)$ as an empirical frequency of individuals who generate signals $s\in E$.  The empirical frequency $\pi$ then generates an empirical frequency of skills, which is $p_{\pi}$.

We say that the set of signals $\T$ is \emph{non-discriminatory} if
for any information structures $\pi,\pi'\in\Delta(\T)$, and any finite set $A\subseteq \Re^{\Theta}$, if $p_{\pi}=p_{\pi'}$, then \[\int_{\T} v_A(t)d\pi(t) = \int_{\T} v_A(t)d\pi'(t).\]

Interpret $\pi(E)$ as the frequency of individuals of type $E\subseteq \T$.  Under the competitive markets assumption, the set $\T$ being non-discriminatory means that the
average  remuneration paid to a class of workers with distribution
$\pi$ ultimately depends only on the distribution of their skills.

\subsection{Motivation and a Phelpsian example}

We start by a simple example to recreate the point made by
\cite{phelps}. It is a minimal example; the simplest we can think
of that delivers the Phelpsian message.  Let
$\Theta=\{\ta_1,\ta_2,\ta_3\}$ be the set of states, and $A=\{
(1,0,0),(0,1/2,3) )\}$ be the set of available actions.  Observe that
with this specification, workers are not ``high'' or ``low'' quality,
but they simply have differing aptitudes for the available actions.

Suppose
that \[ \T = \{(1,0,0),(1/2,1/2,0),(0,1/2,1/2),(0,0,1) \}
\] is the set of signals, or worker types.

Consider two information structures, $\pi$ and $\pi'$, described in
the table below, together with the profit function $v_A$ resulting
from our assumed $\Theta$ and $A$:
\[\begin{array}{c|cccc}
         &  t=(1,0,0) & t=(1/2,1/2,0) & t=(0,1/2,1/2) & t=(0,0,1) \\ \hline
\pi(t)  & 1/3      &     0           &     2/3       & 0 \\
\pi'(t) & 0         &     2/3        &     0           & 1/3 \\
v_A(t) & 1         &    1/2         &     7/4        & 3 \\
\end{array} \]

There are two populations of workers, say A and B. The two populations
differ in the information that the firm obtains about their
skills. The workers might take a
test, as in \cite{phelps}, and the informational content of the test
might be different for the two populations. So As emit signals
about their skills as given by $\pi$, while Bs distribution over
signals is  $\pi'$. Observe that $p_{\pi}=p_{\pi'} = (1/3,1/3,1/3)$,
reflecting that the populations overall have the same skills.

A  worker from group A reveals that she is either good for action
$a_1=(1,0,0)$ or action
$a_2=(0,1/3,3)$. The B worker reveals the same kind of information,
but less efficiently: a signal $t=(1/2,1/2,0)$ tells the employer that
$a_1$ is the optimal choice given the information at hand, but leaves
the employer  with some doubts as to whether $a_2$ may have been the
optimal action. In consequence, we have
 \[
\int_T v_A(t)d\pi(t) = 1/3 +7/6 > 1/3 + 1 = \int_T v_A(t)d\pi'(t).\]
If workers are paid according to the revenues that they contribute to
the firm, as would be the case in a competitive market, then A workers
are paid more than B workers in aggregate. The differences in expected (or
population) remuneration between the two is purely a consequence of
the informational content in their corresponding signal structures.

In our example of Phelpsian statistical discrimination, the two different
information structures have the same mean.  This is a necessary
requirement for the existence of statistical discrimination. It is
important to point out,  however, that skill can {\em always} be
inferred from wages, even when there is discrimination. We present
Proposition~\ref{prop:wage} to make this point.

For a set of actions $A=\{a_1,\ldots,a_n\}$, and action $k$, let
$A+k=\{a_1+k,\ldots,a_n+k\}$.

\begin{proposition}\label{prop:wage}
For any $\T$ and any set of actions $A$, if $\pi,\pi'\in\Delta(\T)$ for which $p_{\pi}\neq p_{\pi'}$, then there is $k$ for which \[\int_T v_{A+k}(t)d\pi(t) \neq \int_T v_{A+k}(t)d\pi'(t).\]
\end{proposition}

\subsection{When discrimination is impossible.}

Our  discussion suggests that discrimination is tied to
identification. Skills are always identified from payoffs, even when there is
discrimination (Proposition~\ref{prop:wage}). The problem is the
converse identification: Here we show that the absence of discrimination is
equivalent to the ability to estimate skills from
signals. Importantly, we show that this can only happen when payments are linear in
signals. So the absence of discrimination is equivalent to the
existence of a state-dependent, signal-independent, ``fair'' payoff. Payments  equal the
expected value of such a payoff, and are called fair valuations.

We say that $\T$
\begin{itemize}
\item is \emph{identified} if for any $\pi,\pi'\in\Delta(\T)$, if
  $p_{\pi}=p_{\pi'}$, then $\pi = \pi'$;
\item \emph{admits fair valuations} if for any finite subset $A\subseteq \Re^{\Theta}$, there is $\alpha_A\in\Re^{\Theta}$ for which for all $t\in \T$, \[v_A(t)=\sum_{\theta}\alpha_A(\theta)t(\theta).\]
\item \emph{admits fair valuations for binary sets} if for any binary subset $A\subseteq \Re^{\Theta}$, there is $\alpha_A\in\Re^{\Theta}$ for which for all $t\in \T$, $v_A(t)=\sum_{\theta}\alpha_A(\theta)t(\theta)$.
\end{itemize}

The notion that $\T$ admits fair valuations captures the idea that any
individual is paid according to her expected skill. Thus, for $A$,
$\alpha_A(\theta)$ represents the value to the firm with technology
$A$ of skill set $\theta\in\Theta$, and if an individual sends signal
$s$ then she is paid the expected value of $\al_A$ according to
$s$. Importantly, if $\pi\in \Delta(\T)$, then \[\int v_A(s)d\pi(s) =
\al_A\cdot \int sd\pi(s) = \al_A\cdot p_\pi.\] So, under fair
valuations,  the expected payment to a population
of agents with information structure $\pi$ only depends on the
distribution of skills in that population.

Finally, say that $\T$ is \emph{non-discriminatory for binary sets} if
for any $\pi,\pi'\in\Delta(\T)$ and any binary $A\subseteq
\Re^{\Theta}$, if $p_{\pi}=p_{\pi'}$, then \[\int_{\T} v_A(t)d\pi(t) =
\int_{\T} v_A(t)d\pi'(t).\]

\begin{theorem}\label{thm:ident}The following are equivalent.
\begin{enumerate}
\item $\T$ is non-discriminatory.
\item $\T$ is non-discriminatory for binary sets.
\item $\T$ is identified.
\item $\T$ admits fair valuations.
\item $\T$ admits fair valuations for binary sets.
\end{enumerate}
\end{theorem}

The main import of Theorem~\ref{thm:ident} is that there is a $\al_A$,
independent of the signal $s$, so that the optimal contribution of the worker to
the firm is the expected value of $\al_A$.  The worker is therefore
remunerated according to some ``fundamental'' value $\al_A$, and
receives the expectation of $\al_A$ according to the signal $s$.

\begin{proposition}\label{prop:LP}If $\T$ admits fair valuations, then
  for each finite $A\subseteq\Re^{\Theta}$ and corresponding
  $\alpha_A\in\Re^{\Theta}$, we have for every $s^*\in\T$:
\[\sum_{\theta}\alpha_A(\theta)s^*(\theta) =
\inf\{\sum_{\theta}y(\theta)s^*(\theta):y\in\Re^{\Theta}\mbox{ and
}v_A(s)\leq \sum_{\theta}y(\theta)s(\theta)\forall s\in \T\}.\]
\end{proposition}

Proposition~\ref{prop:LP} means that the value of a
worker with type $s^*$ to the firm is the minimum expected payment that
guarantees the worker a payoff of at least $v_A(s)$, for all signals
$s\in \T$. This is a kind of participation, or individual rationality,
constraint. The worker may be able to guarantee a payment of $v_A(s)$
on the market, if her signal is $s$, and thus a firm must guarantee at least
$v_A(s)$ in its choice of the ``fair'' payoff $\al_A\in\Re^\Theta$.

\subsection{Connection to Bayesian persuasion}\label{sec:KB} The recent literature
on Bayesian persuasion (\cite{kamenica2011bayesian}) deals with the
optimal design of information structures. It turns out that the value
of optimal information design is linear if and only if $\T$ admits no
discrimination.

We now focus a bit more in depth on the notion of signal structure.
As in \citet{blackwell1953equivalent}, there is a natural notion of
``comparative informativeness'' for $\pi,\pi'\in\Delta(\T)$.  We say
that $\pi$ is \emph{more informative} than $\pi'$ if for every $A$,
$\int v_A(t) d\pi(t) \geq \int v_A(t) d\pi'(t)$.  Most economists will
have heard of the notion of a ``mean-preserving spread;'' $\pi$ turns
out to be more informative than $\pi'$ if it consists of a
mean-preserving spread of $\pi'$.

We know that optimal information design will never utilize signal
structures that are dominated according to the more informativeness
order.  As a result,  optimal information structures will place
probability zero on signals that can be obtained as the mean of other
signals. Formally, an optimal information structure will have support
on  the extreme points of the convex hull of $\T$.


Now, let $T$ be the closed convex hull of $\T$. An information
structure is any probability measure $\pi\in\Delta(T)$. Then
define $W_A:T\rightarrow\Re$ via
\[
W_A(s) \equiv \max\{
\int_{T} v_A(\tilde s) d\pi(\tilde s) : \pi\in\Delta(T)\text{ and }  s
= \int_T \tilde s d\pi(\tilde s)
\}.\] $W_A(s)$ is the value of an optimal information structure for a
population with skill distribution $s$.
In the following, $\partial T$ denotes  the extreme
points of $T$; those points which are not convex combinations of other
points in $T$.

Return to our motivating ``Phelpsian'' example.  There, discrimination
was present even though  $\T$ consisted of the extreme points of its
convex hull $T$, and thus $\T$ was maximally informative. Phelps'
original point can thus be refined: discrimination obtains because an employer
has ``different'' information about two classes of individuals, rather
than ``better'' information.

Let us see how this manifests itself in the choice of optimal
information structure.  In this case, for each $s\in \T$, we have
(clearly) $W_A(s)=v_A(s)$, as each $s$ is extreme in the convex hull
of $T$.  We therefore obtain:
$(2/3)W_A(1/2,1/2,0)+(1/3)W_A(0,0,1)=\frac{4}{3}<\frac{3}{2}=(1/3)v_A(1,0,0)+(2/3)v_A(0,1/2,1/2)\leq
W_A(1/3,1/3,1/3)$.

Hence, $W_A$ is nonlinear in this case.  This is a general artifact of
non-identification and discrimination, as is evidenced by the
following result.

\begin{corollary}\label{cor:KG}
For any $\T$, $\partial T$ is non-discriminatory iff for every $A$,
$W_A$ is affine (linear).\footnote{Because the domain of $W_A$ is a
  set of probability measures, $W_A$ is linear if it is affine. In
  fact, in this case we have $W_A(s) =
  \sum_{\ta\in\Theta}\al_A(\ta)s(\ta)$, where $\al_A$ is
  as in Proposition~\ref{prop:LP}.}
\end{corollary}

As in \citet{kamenica2011bayesian}, $W_A$ is always weakly concave, which
admits the possibility that it is affine.  Corollary~\ref{cor:KG} says
that  discrimination is possible exactly when $W_A$ exhibits strict
concavities.


\section{Conclusion}

We have formulated Phelps' theory of statistical discrimination using
the modern language of information design. Our results shed new light
on the nature of discrimination, and on some of the empirical
approaches one might take to establish the existence of statistical
discrimination.

Statistical discrimination turns out to be equivalent to the absence
of econometric identification of signals from skills. While the
identification of skills from salaries is always possible, even in the
presence of discrimination, we show that the crucial identification
property is that of signals from skills.

In second place, we connect discrimination with the source of worker
remunerations. We show that identification is impossible if and only
if remunerations are linear in ``fair'' skill-dependent,
signal-independent,  payoffs.

Our results have immediate consequences for empirical research on
discrimination. They imply that discrimination is absent if and only
if empirical approaches to linearly estimating fair skills-based
payoffs are viable.

\section{Proofs}

Let $T$ be the closed convex hull of $\T$. Recall that  $\partial T$
denotes the extreme points of $T$.  The definition of $v_A$
extends to $T$.  Let $Y_A:T\rightarrow \Re$ be the concave envelope of
$v_A$, defined as the pointwise infimum of the affine functions that
dominate $v_A$. So if $\cala(T)$ denotes the space of all affine functions
on $T$, then $v_A(t) = \inf\{ l(t): l\in \cala(T)\text{ and } v_A\leq l\}$.
Recall the definition of $W_A$ from Section~\ref{sec:KB}.

\begin{lemma}\label{lem:yw}
$Y_A=W_A$
\end{lemma}

\begin{proof}
  Let $l: T\rightarrow \Re$ be an affine function and $v_A\leq l$.
For any $\pi\in\Delta(T)$ with $\int_T qd\pi(q) = p$,
\[
\int_T v_A(q)d\pi(q) \leq \int_T l(q)d\pi(q)  = l \left( \int_T q
  d\pi(q)\right) = l(p),\] as $l$ is affine. Thus $W_A\leq l$, as $\pi$ was
arbitrary. This implies that $W_A\leq Y_A$, as $l$ was arbitrary.

Now suppose that $W_A(p) < Y_A(p)$. Recall that $W_A$ is concave. Then
the set
$D= \{ (q,y)\in T\times\Re: y\leq W_A(q)  \}$ is closed and convex, so
there exists $\al$ with $(q,y)\cdot \al \leq (p,W_A(p))\cdot \al
< (p,y')\cdot \al$ for all
$(q,y)\in D$ and all $y'\geq Y_A(p)$. Write $\al=(\al^1,\al^2)\in
\Re^\Theta\times\Re$. Clearly we cannot have $\al^2=0$ as
$(p,W_A(p)\in D$. Consider the affine function $l:T\rightarrow \Re$
defined by \[q\mapsto
(1/\al^2) ((p,W_A(p))\cdot \al - \al^1 \cdot q).\]

This means that $l(p)=W_A(p) < Y_A(p)$. Moreover,
for any $q\in T$, $\al\cdot (q,W_A(q)) \leq \al\cdot (p,W_A(p))$;
hence,
\[l(q) =
(1/\al^2) \al^1 \cdot p + W_A(p)  - (1/\al^2) \al^1
\cdot q \geq W_A(q) \geq v_A(q),
\] where the last inequality follows from the definition of
$W_A$. Then $l\in \cala(T)$, $v_A\leq l$, and $l(p)<Y_A(p)$;
a contradiction.
\end{proof}

\subsection{Proof of Theorem~\ref{thm:ident}}

By the Choquet-Meyer Theorem (Theorem II.3.7 in \cite{alfsen2012compact} or p. 56-57 in \citet{phelpschoquet}), $T$ is a simplex
iff $\partial T$ is identified.

Now, to prove the theorem: it is obvious that $3\then 1\then 2$. We
shall prove that $2\then 3$. To this end, let $\T$ be non-discriminatory for
binary menus.
The proof that $2\then 3$ has two parts. The first is to show that
$\T=\partial T$. The second is that $T$ must be a simplex.

First, it is obvious by definition of $T$ that $\partial T\subseteq \T$. So we
prove that  $\T\subseteq \partial T$. To this end, suppose by means of
contradiction that there is $s^*\in \T$ for which there are $t,t'\in
T$, $t\neq t'$, and $\gamma \in (0,1)$ for which $s^* = \gamma t +
(1-\gamma)t'$.
Let $f = (s^*-t)+[t\cdot s^* - s^* \cdot s^*]\mathbf{1}$
and $g=-f$.  Observe that $f \cdot s^* =  0$,
$g\cdot t = -t\cdot (s^*-t) - s^*\cdot (t-s^*) > 0$ and
$f\cdot t' = (s^*-t)\cdot(t'-s^*) = \gamma(1-\gamma)(t'-s^*)\cdot (t'-s^*)>0$.

Let $A \equiv \{f,g\}$. Then we obtain that $v_A(t)\geq g\cdot t>0$,
$v_A(t')\geq f\cdot t'>0$, while $v_A(s^*) = 0$ (as $f \cdot s^*=g
\cdot s^* =  0$).

Now, for each of $t,t'$, there are finitely supported (by
Caratheodory's theorem) $\pi_t$ and $\pi_{t'}$ on $\partial T$ (so in
particular on $\T$) for which $t = \int_{\T}s d\pi(s)$ and $t' =
\int_{\T}s d\pi'(s)$.
This means that
$\int_{\T} v_A(s) d\pi(s) \geq v_A(t)>0$ and
$\int_{\T} v_A(s) d\pi'(s) \geq v_A(t')>0$, as $v_A$ is convex.
Then
\[\int_{\T} v_A(s) d(\g \pi + (1-\g)\pi') (s) >0.\] But this contradicts
2 as $\int_{\T} s d(\g \pi + (1-\g)\pi') (s)  = \g t + (1-\g) t' = s^*$
and $v_A(s^*)=0$.


So we have shown that $\T=\partial T$, and we turn to the proof that
$T$ is a simplex (and thus $\T$ identified). By
\cite{alfsen2012compact} Theorem II.4.1, since $T$ is convex and
compact, $T$ is a simplex if and only if
$\cala(T)$ forms a lattice in the usual (pointwise) ordering on
functions.  So, suppose by means of contradiction that $\cala(T)$ does
not form a lattice.  Then, there are $f,g\in \cala(T)$ which possess
no supremum in $\cala(T)$.

\begin{lemma}\label{lem:obvious} Let $f,g\in\cala(T)$.
For any $z\in\partial T$, if $f(z) \geq g(z)$, then there is $h\in\cala(T)$ for which $h\geq f,g$ and $h(z) = f(z)$.
\end{lemma}

\begin{proof}
Let $M$ be the subgraph of the concave envelope of $v_{\{f,g\}}$.  Observe by definition that it is the convex hull of the points $\{(z,v_{\{f,g\}}):z\in \T\}$, so that it is polyhedral (Corollary 19.I.2 of \citet{rockafellar}).  Therefore, by definition of polyhedral concave function, there is $h$ supporting it at $(z,f(z))$.
\end{proof}

From Lemma~\ref{lem:obvious}, and the fact that $f$ and $g$ possess no
supremum in $\cala(T)$, it follows that there is no affine function
$h$ for which for all $z\in \partial T$, $h(z)= \max\{f(z),g(z)\}$.
Consequently, if $A \equiv \{f,g\}$, then $Y_A$ is not affine, since
for all $z\in \partial T$, it follows that $Y_A(z) =
\max\{f(z),g(z)\} = v_A(z)$.  Now, $Y_A$ being concave and not affine means that there is
$\hat \pi\in\Delta(T)$ with $\int_{\T} Y_A(q)d \hat\pi(q) < Y_A(p_{\hat\pi})$. Since
$\T=\partial T$, and $Y_A$ is concave, we can in fact find (by
Lemma 4.1 in \cite{phelpschoquet})  $\pi\in\Delta(\T)$ with
$p_{\hat\pi} = p_\pi$ and
\[\int_{\T} v_A(q)d\pi(q) =
\int_{\T} Y_A(q)d\pi(q)\leq  \int_{\T} Y_A(q)d\hat \pi(q)< Y_A(p_{\hat\pi})=Y_A(p_\pi),
\] where the first equality follows from $v_A(q)=Y_A(q)$ for $q\in\T$,
and the second inequality from the choice of $\pi$.

Now, by Lemma~\ref{lem:yw}, $Y_A(p_\pi) = \sup\{
\int v_A(q)d \tilde \pi(q) :\tilde\pi\in\Delta(T) \text{ and }
p_{\tilde \pi} = p_\pi\} $. Then there is $\pi'\in\Delta(\T)$ (as the
sup is achieved for a measure with support in $\partial T=\T$) with
$p_\pi = p_{\pi'}$ and
$\int_{\T} v_A(q)d\pi(q) < \int_{\T} v_A(q)d\pi'(q)$, contradicting the fact that $\T$ is non-discriminatory for binary menus.

So, we have shown that $\T$ form the vertices of a simplex.  This establishes that $\T$ is identified.

Now, we establish the equivalence of 3, 4, and 5.  Again, as we already claimed, $T$ is a simplex iff $\cala(T)$ is a lattice.  First, let us show that if $\T$ admits fair valuations for binary sets, then $\cala(T)$ is a lattice.  To see this, let $f,g\in\cala(T)$.  Let $A\equiv \{f,g\}$ and observe that since $T$ is a simplex, there is an affine function $\alpha_A$ such that for each $s\in \T$, we have $\alpha_A \cdot s = \max\{f\cdot s,g\cdot s\}$, since $\T$ admits fair valuations for binary sets.  Clearly, for all $t\in T$, $\alpha_A \cdot t \geq \max\{f\cdot t,g\cdot t\}$.  Suppose $h\in\cala(T)$ with $h\geq f,g$.  Let $t\in T$ be arbitrary, and let $\pi\in\Delta(\T)$ be such that $p_{\pi}=T$.  Then $h(t)=\int_{\T}h(s)d\pi(s)\geq \int_{\T}\alpha_A(s)d\pi(s)=\alpha_A(t)$.  So $\alpha_A$ is the join of $f,g$.  So $T$ is a simplex.  That $\T = \partial T$ can be proved similarly to the above.  So, $\T$ forms the vertices of a simplex, and is hence identified.

Conversely, let us show that if $\T$ is identified, then it admits fair valuations.  So, let $A$ be finite.  View $A$ as a subset of $\cala(T)$.  Let $\alpha_A\in\cala(T)$ be the join of this finite set.  By Lemma~\ref{lem:obvious} and an obvious induction argument, for every $s\in\T$, $\alpha_A\cdot s = \max_{a\in A}a\cdot s=v_A(s)$.  So $A$ admits fair valuations.

\subsection{Proof of Proposition~\ref{prop:LP}}
The Lagrangian for the maximization problem in the definition of $W_A$ is
\begin{align*}
L(\pi,\la) & = \int_{T} v_A(t)d\pi(t) + \la\cdot \left[ p-\int_T
             qd\pi(q)  \right]  \\
& = \la\cdot p + \int_T ( v_A(t) -\la\cdot p )d\pi(t)
\end{align*}
 and apply the maximin theorem (see for example Theorem  6.2.7 in
\cite{aubin2006applied}, which applies here because $\Delta(T)$ is
compact).

\subsection{Proof of Proposition~\ref{prop:wage}}
Observe that for any $A$ and any action $l$, we have $v_{A+l}(t)=v_A(t)+l\cdot t$.  Now, since $p_{\pi}\neq p_{\pi'}$, there is $l$ for which $l\cdot p_{\pi}\neq l\cdot p_{\pi'}$.  Consequently, there is $\alpha$ for which:
\[\alpha l \cdot (p_{\pi}-p_{\pi'})\neq \int_T v_A(t)d\pi'(t)-\int_T v_A(t)d\pi(t).\]

Let $k= \alpha l$, and conclude that:
\[\int_T v_{A+k}(t)d\pi(t)=k\cdot p_{\pi}+\int_T v_A(t)d\pi(t) \neq k\cdot p_{\pi'}+\int_T v_A(t)d\pi'(t)=\int_T v_{A+k}(t)d\pi'(t).\]

\subsection{Proof of Corollary~\ref{cor:KG}}

By the  Choquet-Meyers Theorem (Theorem II.3.7 in \cite{alfsen2012compact}) $T$ is a
simplex iff the concave envelope of every lower semicontinuous and
convex function is affine.  Clearly, when $\T$ is identified, $T$ is a simplex, and since $v_A$ is convex and lower semicontinuous, we obtain that $W_A=Y_A$, the concave envelope.  So $W_A$ is affine.

Conversely, suppose that $W_A$ is affine for each finite $A$.  We will show that $T$ is a simplex (so that $\partial T$ forms the vertices of a simplex, and is identified).  But this again follows from the fact that $W_A$ is the smallest concave function on $T$ dominating each $a\in A$.  Since it is affine, it follows that $\cala(T)$ is a lattice, and hence $T$ is a simplex.

\bibliographystyle{ecta}
\bibliography{disc}

\begin{thebibliography}{12}
\newcommand{\enquote}[1]{``#1''}
\expandafter\ifx\csname natexlab\endcsname\relax\def\natexlab#1{#1}\fi

\bibitem[\protect\citeauthoryear{Aigner and Cain}{Aigner and
  Cain}{1977}]{aigner1977statistical}
\textsc{Aigner, D.~J. and G.~G. Cain} (1977): \enquote{Statistical theories of
  discrimination in labor markets,} \emph{Industrial and Labor Relations
  Review}, 30, 175--187.

\bibitem[\protect\citeauthoryear{Alfsen}{Alfsen}{2012}]{alfsen2012compact}
\textsc{Alfsen, E.~M.} (2012): \emph{Compact convex sets and boundary
  integrals}, vol.~57, Springer Science \& Business Media.

\bibitem[\protect\citeauthoryear{Arrow}{Arrow}{1971}]{arrow71}
\textsc{Arrow, K.~J.} (1971): \enquote{Some models of Racial Discrimination in
  the Labor Market,} Tech. Rep. RM-6253-RC, RAND.

\bibitem[\protect\citeauthoryear{Arrow}{Arrow}{1973}]{arrow73}
---\hspace{-.1pt}---\hspace{-.1pt}--- (1973): \emph{The Theory of
  Discrimination. S. 3--33 in: O. Ashenfelter/A. Rees (Hrsg.), Discrimination
  in Labor Markets}, Princeton University Press.

\bibitem[\protect\citeauthoryear{Aubin and Ekeland}{Aubin and
  Ekeland}{2006}]{aubin2006applied}
\textsc{Aubin, J.-P. and I.~Ekeland} (2006): \emph{Applied nonlinear analysis},
  Courier Corporation.

\bibitem[\protect\citeauthoryear{Becker}{Becker}{1957}]{becker}
\textsc{Becker, G.~S.} (1957): \emph{The Theory of Discrimination}, University
  of Chicago Press.

\bibitem[\protect\citeauthoryear{Blackwell}{Blackwell}{1953}]{blackwell1953equivalent}
\textsc{Blackwell, D.} (1953): \enquote{Equivalent comparisons of experiments,}
  \emph{The annals of mathematical statistics}, 265--272.

\bibitem[\protect\citeauthoryear{Kamenica and Gentzkow}{Kamenica and
  Gentzkow}{2011}]{kamenica2011bayesian}
\textsc{Kamenica, E. and M.~Gentzkow} (2011): \enquote{Bayesian persuasion,}
  \emph{American Economic Review}, 101, 2590--2615.

\bibitem[\protect\citeauthoryear{Machina}{Machina}{1984}]{machina}
\textsc{Machina, M.~J.} (1984): \enquote{Temporal risk and the nature of
  induced preferences,} \emph{Journal of Economic Theory}, 33, 199--231.

\bibitem[\protect\citeauthoryear{Phelps}{Phelps}{1972}]{phelps}
\textsc{Phelps, E.~S.} (1972): \enquote{The Statistical Theory of Racism and
  Sexism,} \emph{American Economic Review}, 62, 659--661.

\bibitem[\protect\citeauthoryear{Phelps}{Phelps}{2000}]{phelpschoquet}
\textsc{Phelps, R.~R.} (2000): \emph{Lectures on Choquet's theorem, second
  edition}, Springer Science \& Business Media.

\bibitem[\protect\citeauthoryear{Rockafellar}{Rockafellar}{1970}]{rockafellar}
\textsc{Rockafellar, R.~T.} (1970): \emph{Convex analysis}, Princeton
  university press.

\end{thebibliography}

\end{document}